\documentclass[envcountsame]{llncs}

\usepackage{epsfig}
\usepackage{latexsym}
\usepackage{amsmath}
\usepackage{amsfonts}
\usepackage{amssymb}
\usepackage{graphicx}
\usepackage{algorithm}
\usepackage{myalgorithmic}
\usepackage{enumitem}
\usepackage{changepage}
\usepackage{booktabs}
\usepackage{multirow}
\usepackage{array}
\usepackage{tikz}
\usepackage{float}
\usepackage{sectsty}
\usepackage[mathscr]{euscript}
\usepackage{color}

\chapterfont{\raggedright}

\newfloat{algorithm}{t}{lop}
\newcolumntype{C}[1]{>{\centering\arraybackslash}p{#1}}

\usetikzlibrary{arrows,shapes,automata,backgrounds,petri}
\usetikzlibrary{positioning,calc}
\tikzstyle{legendborder}=[rectangle, draw, black, rounded corners, thin, top color=white, text=black, minimum width=2.5cm]
\tikzstyle{legendnoborder}=[rectangle, draw, white, rounded corners, thin, top color=white, text=black, minimum width=2.5cm]
\tikzstyle{selected edge} = [draw,line width=1pt,purple]

\newcommand{\mc}{\mathcal}
\newcommand{\Oh}{{\mc{O}}}
\newcommand{\Ohstar}{\Oh^\star}

\renewcommand\algorithmicthen{}
\renewcommand\algorithmicdo{}

\newtheorem{clm}{Claim}

\makeatletter
\newcommand\xleftrightarrow[2][]{%
  \ext@arrow 9999{\longleftrightarrowfill@}{#1}{#2}}
\newcommand\longleftrightarrowfill@{%
  \arrowfill@\leftarrow\relbar\rightarrow}
\makeatother

\def\ruler{\leftrightarrow}
\def\reach{\leftrightarrow^*}
\def\tapeL{\mbox{\textdollar}}
\def\tapeR{\mbox{\rlap c{\footnotesize/}}}
\def\tapeE{\mbox{\textvisiblespace}\ }

\def\sig{\tau}
\def\extsig{\pi}
\def\table{A}
\def\bag{X}
\def\used{\textnormal{\texttt{used}}}
\def\unused{\textnormal{\texttt{unused}}}
\def\ttleft{\textnormal{\texttt{left}}}
\def\ttright{\textnormal{\texttt{right}}}

\begin{document}
\title{Reconfiguration over tree decompositions}
\author
{
    Amer E. Mouawad\inst{1}\thanks{Research supported by the Natural Science and Engineering Research Council of Canada.}
    \and Naomi Nishimura\inst{1}$^{\star}$
    %\and Vinayak Pathak\inst{1}$^{\star}$
    \and Venkatesh Raman\inst{2}
    \and \\Marcin Wrochna\inst{3}
}
\institute
{
    David R. Cheriton School of Computer Science\\
    University of Waterloo, Ontario, Canada.\\
    \email{\{aabdomou, nishi\}@uwaterloo.ca}
    \and
    Institute of Mathematical Sciences\\
    Chennai, India.\\
    \email{vraman@imsc.res.in}\\
    \and
    Institute of Computer Science\\
    Uniwersytet Warszawski, Warsaw, Poland.\\
    \email{mw290715@students.mimuw.edu.pl}
}
\maketitle
\sloppy

\begin{abstract}
A vertex-subset graph problem \textsc{$Q$} defines which subsets
of the vertices of an input graph are feasible solutions.
The reconfiguration version of a vertex-subset problem \textsc{$Q$} asks
whether it is possible to transform one feasible solution for \textsc{$Q$}
into another in at most $\ell$ steps, where each step is a vertex addition or deletion, and
each intermediate set is also a feasible solution for \textsc{$Q$}
of size bounded by $k$.
Motivated by recent results establishing W[1]-hardness of the
reconfiguration versions of most vertex-subset problems parameterized by $\ell$, we investigate
the complexity of such problems restricted to graphs of bounded treewidth.
We show that the reconfiguration versions of
most vertex-subset problems remain PSPACE-complete on graphs of treewidth at most $t$ but
are fixed-parameter tractable parameterized by $\ell + t$ for all vertex-subset problems definable
in monadic second-order logic (MSOL). To prove the latter result, we introduce a technique
which allows us to circumvent cardinality constraints and define
reconfiguration problems in MSOL.
\end{abstract}

\section{Introduction}
{\em Reconfiguration problems} allow the study of structural and algorithmic
questions related to the solution space of computational problems,
represented as a {\em reconfiguration graph} where feasible solutions
are represented by nodes and adjacency by edges
\cite{CVJ08,GKMP09,IDHPSUU11}; a path is equivalent to
the step-by-step transformation of one solution into another
as a {\em reconfiguration sequence} of {\em reconfiguration steps}.

Reconfiguration problems have so far been studied mainly under
classical complexity assumptions, with most work devoted to
deciding whether it is possible to find a path between two solutions.
For several problems, this question has been shown to be
PSPACE-complete~\cite{B12,IDHPSUU11,IKD12}, using reductions that 
construct examples where the length $\ell$ of reconfiguration
sequences can be exponential in the size of the input graph.
It is therefore natural to ask whether we can achieve tractability if we
allow the running time to depend on $\ell$ or on
other properties of the problem, such as a bound $k$ on the size of feasible solutions.
These results motivated Mouawad et al.~\cite{MNRSS13} to study reconfiguration under the
{\em parameterized complexity} framework~\cite{DF97}, showing the
W[1]-hardness of \textsc{Vertex Cover Reconfiguration (VC-R)},
\textsc{Feedback Vertex Set Reconfiguration (FVS-R)}, and
\textsc{Odd Cycle Transversal Reconfiguration (OCT-R)} parameterized by $\ell$,
and of \textsc{Independent Set Reconfiguration (IS-R)},
\textsc{Induced Forest Reconfiguration (IF-R)}, and
\textsc{Induced Bipartite Subgraph Reconfiguration (IBS-R)}
parameterized by $k + \ell$~\cite{MNRSS13}.

Here we focus on reconfiguration problems restricted
to $\mathscr{C}_t$, the class of graphs of treewidth at most $t$.
In Section~\ref{sec-hardness}, we show that a large number of reconfiguration problems,
including the six aforementioned problems, remain PSPACE-complete
on $\mathscr{C}_t$, answering a question left open by Bonsma~\cite{B12Planar}.
The result is in fact stronger in that it applies to graphs of bounded
bandwidth and even to the question of finding a reconfiguration sequence of \emph{any} length.

In Section~\ref{sec-meta}, using an adaptation of Courcelle's
cornerstone result~\cite{Courcelle90}, we present a meta-theorem
proving that the reconfiguration versions of all vertex-subset problems definable
in monadic second-order logic become tractable on $\mathscr{C}_t$
when parameterized by $\ell + t$.
Since the running times implied by our meta-theorem are far from practical, we
consider the reconfiguration versions of problems defined in terms of
hereditary graph properties in Section~\ref{sec-dp}. In particular,
we first introduce signatures to succinctly represent reconfiguration
sequences and define ``generic'' procedures on signatures which can be used to
exploit the structure of nice tree decompositions.
We use these procedures in Section~\ref{sec-vc-tw} to design algorithms solving
\textsc{VC-R} and \textsc{IS-R} in $\Ohstar(4^{\ell} (t+3)^{\ell})$ time
(the $\Ohstar$ notation suppresses factors polynomial in $n$, $\ell$, and $t$).
In Section~\ref{sec-octr-fvsr-tw}, we extend the algorithms to solve
\textsc{OCT-R} and \textsc{IBS-R} in $\Ohstar(2^{\ell t} 4^\ell (t+3)^{\ell})$ time, as
well as \textsc{FVS-R} and \textsc{IF-R} in $\Ohstar(t^{\ell t} 4^\ell (t+3)^{\ell})$ time.
We further demonstrate in Section~\ref{sec-vcr-shifting} that \textsc{VC-R} and \textsc{IS-R} parameterized by
$\ell$ can be solved in $\Ohstar(4^{\ell} (3\ell + 2)^{\ell})$
time on planar graphs by an adaptation of Baker's shifting technique~\cite{BAKER94}.

\section{Preliminaries}\label{sec-prelim}
For general graph theoretic definitions, we refer
the reader to the book of Diestel~\cite{D05}.
We assume that each input graph $G$ is a
simple undirected graph with vertex set $V(G)$ and
edge set $E(G)$, where $|V(G)| = n$ and $|E(G)| = m$.
The {\em open neighborhood} of a vertex $v$ is denoted by $N_G(v) = \{u \mid uv \in E(G)\}$ and the
{\em closed neighborhood} by $N_G[v] = N_G(v) \cup \{v\}$.
For a set of vertices $S \subseteq V(G)$,
we define $N_G(S) = \{v \not\in S \mid uv \in E(G), u \in S\}$ and $N_G[S] = N_G(S) \cup S$.
We drop the subscript $G$ when clear from context.
The subgraph of $G$ induced by $S$ is denoted by $G[S]$, where
$G[S]$ has vertex set $S$ and edge set $\{uv \in E(G) \mid u, v \in S\}$.
Given two sets $S_1,S_2 \subseteq V(G)$, we let
$S_1 \Delta S_2 = \{S_1 \setminus S_2\} \cup \{S_2 \setminus S_1\}$
denote the symmetric  difference of $S_1$ and $S_2$.

We say a graph problem \textsc{$Q$} is a {\em vertex-subset} problem whenever
feasible solutions for \textsc{$Q$} on input $G$ correspond to subsets of $V(G)$.
\textsc{$Q$} is a {\em vertex-subset minimization (maximization)} problem whenever
feasible solutions for \textsc{$Q$} correspond to subsets of $V(G)$ of size at most (at least) $k$, for
some integer $k$. The {\em reconfiguration graph} of a vertex-subset
minimization (maximization) problem \textsc{$Q$}, $R_{\textsc{min}}(G,k)$ ($R_{\textsc{max}}(G,k)$),
has a node for each $S \subseteq V(G)$ such that $|S| \leq k$
($|S| \geq k$) and $S$ is a feasible solution for \textsc{$Q$}.
We say $k$ is the {\em maximum (minimum) allowed capacity} for
$R_{\textsc{min}}(G,k)$ ($R_{\textsc{max}}(G,k)$).
Nodes in a reconfiguration graph are adjacent if they
differ by the addition or deletion of a single vertex.

\begin{definition}
For any vertex-subset problem \textsc{$Q$}, graph $G$, positive integers $k$ and $\ell$, $S_s \subseteq V(G)$,
and $S_t \subseteq V(G)$, we define four decision problems:
\begin{itemize}
\item \textsc{$Q$-Min$(G, k)$}: Is there $S \subseteq V(G)$ such that $|S| \leq k$ and $S$ is a feasible solution for $Q$?
\item \textsc{$Q$-Max$(G, k)$}: Is there $S \subseteq V(G)$ such that $|S| \geq k$ and $S$ is a feasible solution for $Q$?
\item \textsc{$Q$-Min-R$(G, S_s, S_t, k, \ell)$}: For $S_s,S_t \in V(R_{\textsc{min}}(G,k))$, is there a path of length at most $\ell$ between the nodes for $S_s$ and $S_t$ in $R_{\textsc{min}}(G,k)$?
\item \textsc{$Q$-Max-R$(G, S_s, S_t, k, \ell)$}: For $S_s,S_t \in V(R_{\textsc{max}}(G,k))$, is there a path of length at most $\ell$ between the nodes for $S_s$ and $S_t$ in $R_{\textsc{max}}(G,k)$?
\end{itemize}
\end{definition}

For ease of description, we present our positive results for
paths of length exactly $\ell$, as all our algorithmic techniques
can be generalized to shorter paths.
Throughout, we implicitly consider reconfiguration problems
as parameterized problems with $\ell$ as the parameter.
The reader is referred to the books of Downey and Fellows~\cite{DF97},
Flum and Grohe~\cite{FG06}, and Niedermeier~\cite{N06}
for more on parameterized complexity.

In Section~\ref{sec-dp}, we consider problems that can be defined
using graph properties, where a {\em graph property} $\Pi$ is a collection of
graphs closed under isomorphism, and is
{\em non-trivial} if it is non-empty and does not contain all graphs.
A graph property is {\em polynomially decidable} if
for any graph $G$, it can be decided in
polynomial time whether $G$ is in $\Pi$.
The property $\Pi$ is {\em hereditary} if for any
$G \in \Pi$, any induced subgraph of $G$ is also in $\Pi$.
For a graph property $\Pi$, $R_{\textsc{max}}(G,k)$ has a node for each
$S \subseteq V(G)$ such that $|S| \ge k$ and $G[S]$ has property $\Pi$, and
$R_{\textsc{min}}(G,k)$ has a node for each $S \subseteq V(G)$
such that $|S| \le k$ and $G[V(G) \setminus S]$ has property $\Pi$.
We use \textsc{$\Pi$-Min-R} and \textsc{$\Pi$-Max-R} instead of \textsc{$Q$-Min-R} and \textsc{$Q$-Max-R}, respectively,
to denote {\em reconfiguration problems for $\Pi$}; examples include \textsc{VC-R},
\textsc{FVS-R}, and \textsc{OCT-R} for the former and \textsc{IS-R},
\textsc{IF-R}, and \textsc{IBS-R} for the latter, for $\Pi$ defined as
the collection of all edgeless graphs, forests, and bipartite graphs,
respectively.

Proofs of propositions, lemmas, and theorems marked with a star can be found in the appendix.

\begin{proposition}\label{prop-dual}
Given $\Pi$ and a collection of graphs $\mathscr{C}$,
if \textsc{$\Pi$-Min-R} parameterized by $\ell$ is fixed-parameter tractable
on $\mathscr{C}$ then so is \textsc{$\Pi$-Max-R}.
\end{proposition}

\begin{proof}
Given an instance $(G, S_s, S_t, k, \ell)$ of \textsc{$\Pi$-Max-R}, where
$G \in \mathscr{C}$, we solve the \textsc{$\Pi$-Min-R} instance
$(G, V(G) \setminus S_s, V(G) \setminus S_t, n - k, \ell)$. Note that the
parameter $\ell$ remains unchanged.

It is not hard to see that
there exists a path between the nodes corresponding
to $S_s$ and $S_t$ in $R_{\textsc{max}}(G,k)$ if and only if
there exists a path of the same length between the nodes corresponding to
$V(G) \setminus S_s$ and $V(G) \setminus S_t$ in $R_{\textsc{min}}(G,n-k)$.
\qed
\end{proof}

We obtain our results by solving \textsc{$\Pi$-Min-R},
which by Proposition~\ref{prop-dual} implies results for \textsc{$\Pi$-Max-R}.
We always assume $\Pi$ to be
non-trivial, polynomially decidable, and hereditary.

Our algorithms rely on dynamic
programming over graphs of bounded treewidth.
A {\em tree decomposition} of a graph $G$ is a pair
$\mathcal{T} = (T, \chi)$, where $T$ is a tree and $\chi$ is
a mapping that assigns to each node $i \in V(T)$ a vertex subset
$X_i$ (called a {\em bag}) such that: (1) $\bigcup_{i \in V(T)}{X_i} = V(G)$,
(2) for every edge $uv \in E(G)$, there exists a
node $i \in V(T)$ such that the bag $\chi(i) = X_i$ contains both $u$ and $v$, and
(3) for every $v \in V(G)$, the set $\{i \in V(T) \mid v \in X_i\}$
forms a connected subgraph (subtree) of $T$.
The {\em width} of any tree decomposition $\mathcal{T}$ is equal
to $\max_{i \in V(T)}|X_i| - 1$.
The {\em treewidth} of a graph $G$, $tw(G)$, is the minimum width of
a tree decomposition of $G$.

For any graph of treewidth $t$, we can compute a tree decomposition of width $t$
and transform it into a nice tree decomposition
of the same width in linear time~\cite{K94}, where a rooted tree
decomposition $\mathcal{T} = (T, \chi)$ with root $root$ of a graph
$G$ is a {\em nice tree decomposition} if each of its nodes is either
(1) a leaf node (a node $i$ with $|X_i| = 1$ and no children), (2) an
introduce node (a node $i$ with exactly one child $j$ such
that $X_i = X_j \cup \{v\}$ for some vertex $v \not\in X_j$;
$v$ is said to be {\em introduced} in $i$), (3) a forget node
(a node $i$ with exactly one child $j$
such that $X_i = X_j \setminus \{v\}$ for some vertex $v \in X_j$;
$v$ is said to be {\em forgotten} in $i$), or (4) a
join node (a node $i$ with two children $p$ and $q$ such that $X_i = X_p = X_q$).
For node $i \in V(T)$, we use $T_i$ to denote the subtree of $T$ rooted at $i$ and
$V_i$ to denote the set of vertices of $G$ contained in the bags of $T_i$.
Thus $G[V_{root}] = G$.

\section{PSPACE-completeness}\label{sec-hardness}
We define a simple intermediary problem that highlights the essential
elements of a PSPACE-hard reconfiguration problem.
Given a pair $H=(\Sigma,E)$, where $\Sigma$ is an alphabet and
$E\subseteq \Sigma^2$ a binary relation between symbols, we say that
a word over $\Sigma$ is an {\em $H$-word} if every two consecutive
symbols are in the relation. If one looks at $H$ as a digraph (possibly with loops),
a word is an $H$-word if and only if it is a walk in $H$.
The \textsc{$H$-Word Reconfiguration} problem asks whether two
given $H$-words of equal length can be transformed into one
another (in any number of steps) by changing one symbol at a time so
that all intermediary steps are also $H$-words.

A {\em Thue system} is a pair $(\Sigma,R)$, where $\Sigma$ is a finite alphabet and
$R\subseteq \Sigma^* \times \Sigma^*$ is a set of rules.
A rule can be applied to a word by replacing one subword
by the other, that is, for two words $s,t\in\Sigma^*$, we
write $s\ruler_R t$ if there is a rule $\{\alpha,\beta\}\in R$
and words $u,v\in\Sigma^*$ such that $s=u\alpha v$ and $t=u\beta v$.
The reflexive transitive closure of this relation defines an equivalence
relation $\reach_R$, where words $s,t$ are equivalent if and only if one can be reached
from the other by repeated application of rules. The {\em word problem} of
$R$ is the problem of deciding, given two words $s,t\in\Sigma^*$, whether
$s\reach_R t$. A Thue system is called {\em $c$-balanced} if for each
$\{\alpha,\beta\}\in R$ we have $|\alpha|=|\beta|=c$.
The following fact is a folklore variant~\cite{BF84} of the classic
proof of undecidability for general Thue systems~\cite{P47}.

\begin{lemma}[*]\label{thue}
There exists a 2-balanced Thue system whose word problem is PSPACE-complete.%(under $\leq^m_{log}$-reducibility).
\end{lemma}

A simple but technical reduction from Lemma~\ref{thue} allows us to show the PSPACE-completeness of \textsc{$H$-Word Reconfiguration}.
The simplicity of the problem
statement allows for easy reductions to various
reconfiguration problems, as exemplified in Theorem~\ref{th-hardness}.
Similar reductions apply to the reconfiguration versions
of, e.g., \textsc{$k$-Coloring}~\cite{CVJ11} and \textsc{Shortest Path}~\cite{KMM11}  --
a comprehensive discussion is available in an online manuscript by the fourth author~\cite{W14}.

\begin{lemma}[*]\label{hWordReconfiguration}
There exists a digraph $H$ for which \textsc{$H$-Word Reconfiguration} is PSPACE-complete.
\end{lemma}

\begin{theorem}\label{th-hardness}
There exists an integer $b$ such that \textsc{VC-R}, \textsc{FVS-R}, \textsc{OCT-R},
\textsc{IS-R}, \textsc{IF-R}, and \textsc{IBS-R} are PSPACE-complete
even when restricted to graphs of treewidth at most $b$.
\end{theorem}

\begin{proof}
Let $H=(\Sigma,R)$ be the digraph obtained from Lemma~\ref{hWordReconfiguration}.
We show a reduction from \textsc{$H$-Word Reconfiguration} to \textsc{VC-R}.

For an integer $n$, we define $G_n$ as follows. The vertex set contains
vertices $v_i^a$ for all $i\in\{1,\dots,n\}$ and $a\in\Sigma$.
Let $V_i = \{v_i^a \mid a\in\Sigma\}$ for $i\in\{1,\dots,n\}$. The edge set of
$G_n$ contains an edge between every two vertices of $V_i$ for $i\in\{1,\dots,n\}$ and
an edge $v_i^a v_{i+1}^b$ for all $(a,b)\not\in R$ and $i\in\{1,\dots,n-1\}$.
The sets $V_i\cup V_{i+1}$ give a tree decomposition of width $b=2|\Sigma|$.

Let $k=n\cdot (|\Sigma|-1)$ and consider a vertex cover $S$ of $G_n$ of size $k$.
For all $i$, since $G_n[V_i]$ is a clique, $S$ contains all vertices of $V_i$ except at most one.
Since $|S|=\sum_i (|V_i| -1)$, $S$ contains all vertices
except exactly one from each set $V_i$, say $v^{s_i}_i$ for some $s_i \in \Sigma$.
Now $s_1 \dots s_n$ is an $H$-word ($s_i s_{i+1} \in R$, as otherwise $v^{s_i}_i v^{s_{i+1}}_{i+1}$
would be an uncovered edge) and any $H$-word can be
obtained in a similar way, giving a bijection between vertex
covers of $G_n$ of size $k$ and $H$-words of length $n$.

Consider an instance $s,t\in\Sigma^*$ of \textsc{$H$-Word Reconfiguration}.
We construct the instance $(G_n,S_s,S_t,k+1,\ell)$ of \textsc{VC-R}, where
$n=|s|=|t|$, $\ell = 2^{n|\Sigma|}$ (that is, we ask for a reconfiguration sequence
of any length) and $S_s$ and $S_t$ are the vertex
covers of size $k$ that correspond to $s$ and $t$, respectively.
Any reconfiguration sequence between
such vertex covers starts by adding a vertex (since $G_n$ has no vertex cover of size
$k-1$) and then removing another (since vertex covers larger than $k+1$ are
not allowed), which corresponds to changing
one symbol of an $H$-word. This gives a one-to-one correspondence between
reconfiguration sequences of $H$-words and reconfiguration sequences (of exactly twice the length) between
vertex covers of size $k$.
The instances are thus equivalent.

This proof can be adapted to \textsc{FVS-R} and \textsc{OCT-R} by replacing
edges with cycles, e.g. triangles~\cite{MNRSS13}.
For \textsc{IS-R}, \textsc{IF-R}, and \textsc{IBS-R}, we simply need to consider
set complements of solutions for \textsc{VC-R}, \textsc{FVS-R}, and \textsc{OCT-R}, respectively.
\qed
\end{proof}

\section{A meta-theorem}\label{sec-meta}
In contrast to Theorem~\ref{th-hardness}, in this section we show that
a host of reconfiguration problems definable
in monadic second-order logic (MSOL)
become fixed-parameter tractable
when parameterized by $\ell + t$.
First, we briefly review the syntax and semantics of
MSOL over graphs. The reader is referred to
the excellent survey by Martin Grohe~\cite{Grohe07} for more details.

We have an infinite set of {\em individual variables},
denoted by lowercase letters $x$, $y$, and $z$, and an infinite set
of {\em set variables}, denoted by uppercase letters $X$, $Y$, and $Z$.
A {\em monadic second-order formula (MSOL-formula)} $\phi$ over a graph $G$ is constructed from {\em atomic formulas}
$\mathcal{E}(x,y)$, $x \in X$, and $x = y$ using the usual Boolean
connectives as well as existential and universal quantification over individual and set variables.
We write $\phi(x_1,\dots,x_r,X_1,\dots,X_s)$ to indicate that $\phi$ is a
formula with free variables $x_1,\dots,x_r$ and $X_1,\dots,X_s$, where
free variables are variables not bound by quantifiers.

For a formula $\phi(x_1,\dots, x_r,X_1,\dots, X_s)$, a graph $G$, vertices
$v_1,\dots,v_r$, and sets $V_1,\dots,V_r$, we write $G \models \phi(v_1,\dots,v_r,V_1,\dots,V_r)$
if $\phi$ is satisfied in $G$ when $\mathcal{E}$ is interpreted
by the adjacency relation $E(G)$, the variables $x_i$ are interpreted by $v_i$, and variables $X_i$ are interpreted by $V_i$.
We say that a vertex-subset problem $Q$ is
{\em definable in monadic second-order logic} if
there exists an MSOL-formula $\phi(X)$ with one free set variable
such that $S\subseteq V(G)$ is a feasible solution of problem $Q$ for
instance $G$ if and only if $G \models \phi(S)$.
For example, an independent set is definable by the
formula $\phi_{\textsc{is}}(X)=\forall_x \forall_y (x\in X \wedge y\in X) \to \neg \mathcal{E}(x,y)$.

\begin{theorem}[Courcelle~\cite{Courcelle90}]\label{th-cour}
There is an algorithm that given a MSOL-formula
$\phi(x_1,\dots,x_r,X_1,\dots,X_s)$, a graph $G$, vertices $v_1,\dots,v_r\in V(G)$, and
sets $V_1,\dots,V_s\subseteq V(G)$ decides whether
$G \models \phi(v_1,\dots,v_r,V_1,\dots,V_s)$ in $\Oh(f(tw(G), |\phi|) \cdot n)$ time, for some computable function $f$.
\end{theorem}

\begin{theorem}\label{th-meta}
If a vertex-subset problem $Q$ is definable in monadic
second-order logic by a formula $\phi(X)$, then
\textsc{$Q$-Min-R} and \textsc{$Q$-Max-R}
parameterized by $\ell + tw(G) + |\phi|$ are fixed-parameter tractable.
\end{theorem}

\begin{proof}
We provide a proof for \textsc{$Q$-Min-R} as the proof for \textsc{$Q$-Max-R}
is analogous. Given an instance $(G, S_s, S_t, k, \ell)$ of \textsc{$Q$-Min-R},
we build an MSOL-formula $\omega(X_0,X_\ell)$ such that
$G \models \omega(S_s,S_t)$ if and only if the corresponding
instance is a yes-instance.
Since the size of $\omega$ will be bounded by a function of $\ell + |\phi|$, the statement will follow from Theorem~\ref{th-cour}.

As MSOL does not allow cardinality constraints, we overcome this
limitation using the following technique.
We let $L \subseteq \{-1,+1\}^\ell$ be the set of
all sequences of length $\ell$ over $\{-1,+1\}$ which do not violate the maximum allowed
capacity. In other words, given $S_s$ and $k$, a sequence $\sigma$ is in $L$ if and only if for all $\ell'\leq \ell$ it satisfies
$|S_s| + \sum_{i = 1}^{\ell'}{\sigma[i]} \leq k$, where $\sigma[i]$ is the
$i^{th}$ element in sequence $\sigma$.
We let $\omega = \bigvee_{\sigma\in L} \omega_\sigma$ and
$$ \omega_\sigma(X_0,X_\ell) = \exists_{X_1,\ldots,X_{\ell-1}}\
\bigwedge_{0 \leq i \leq \ell} \phi(X_i) \wedge
\bigwedge_{1 \leq i \leq \ell} \psi_{\sigma[i]}(X_{i-1},X_{i}) $$
where $\psi_{-1}(X_{i-1},X_{i})$ means $X_i$ is obtained from $X_{i-1}$
by removing one element and $\psi_{+1}(X_{i-1},X_{i})$ means it is obtained by adding one element.
Formally, we have:
$$\psi_{-1}(X_{i-1},X_{i}) =
    \exists_x\ x \in X_{i-1} \wedge x \not\in X_{i}
     \wedge \forall y  \left(y \in X_{i} \leftrightarrow (y \in X_{i - 1} \wedge y \neq x)\right)$$
$$\psi_{+1}(X_{i-1},X_{i}) =
    \exists_x\  x \not\in X_{i-1} \wedge x \in X_{i}
    \wedge \forall y \left(y \in X_{i} \leftrightarrow (y \in X_{i - 1} \vee y = x)\right)$$
It is easy to see that $G\models \omega_\sigma(S_s,S_t)$ if and only if there
is a reconfiguration sequence from $S_s$ to $S_t$ (corresponding to $X_0,X_1,\dots,X_\ell$)
such that the $i^{th}$ step removes a vertex if $\sigma[i]=-1$ and adds a vertex if $\sigma[i]=+1$.
Since $|L| \leq 2^\ell$, the size of the MSOL-formula $\omega$ is bounded by an (exponential) function of $\ell + |\phi|$.
\qed
\end{proof}

\section{Dynamic programming algorithms}\label{sec-dp}
Throughout this section we will consider one fixed instance
$(G, S_s, S_t, k, \ell)$ of \textsc{$\Pi$-Min-R} and a nice tree decomposition
$\mathcal{T} = (T, \chi)$ of $G$.
Moreover, similarly to the previous section, we will ask, for a fixed sequence $\sigma \in \{-1,+1\}^\ell$,
whether $G\models \omega_\sigma(S_s,S_t)$ holds.
That is, we ask whether there is a reconfiguration sequence which at 
the $i^{th}$ step removes a vertex when $\sigma[i]=-1$ and adds a vertex when $\sigma[i]=+1$.
The final algorithm then asks such a question for every sequence $\sigma$ which does not violate the maximum allowed
capacity: $|S_s| + \sum_{i = 1}^{\ell'}{\sigma[i]} \leq k$ for all $\ell' \leq \ell$.
This will add a factor of at most $2^\ell$ to the running time.

\subsection{Signatures as equivalence classes}\label{sec-sign}
A reconfiguration sequence can be described as a sequence of steps, each step
specifying which vertex is being removed or added.
To obtain a more succinct representation,
we observe that in order to propagate information up from the leaves
to the root of a nice tree decomposition,
we can ignore vertices outside of the currently considered bag ($X_i$)
and only indicate whether a step has been used by a vertex
in any previously processed bags, i.e. a vertex in $V_i \setminus X_i$.

\begin{definition}\label{def-sign}
A {\em signature} $\sig$ over a set $X \subseteq V(G)$ is a sequence
of steps $\sig[1],\dots,\sig[\ell]\in X \cup \textnormal{\{\used,\unused\}}$.
Steps from $X$ are called \emph{vertex steps}.
\end{definition}

The total number of signatures over a bag $X$ of at most $t$ vertices is $(t+3)^\ell$.
Our dynamic programming algorithms start by considering a signature
with only $\unused$ steps in each leaf node, specify when a vertex may be
added/removed in introduce nodes by replacing $\unused$ steps with
vertex steps ($\sig[i]=\unused$ becomes $\sig[i]=v$ for the introduced vertex $v$),
merge signatures in join nodes, and replace vertex steps with $\used$ steps in forget nodes.

For a set $S \subseteq V(G)$ and a bag $X$, we let $\sig(i,S) \subseteq S \cup X$ denote the set of
vertices obtained after executing the first $i$ steps of $\sig$:
the $i^{th}$ step adds $\sig[i]$ if $\sig[i]\in X$ and $\sigma[i]=+1$, removes
it if $\sig[i]\in X$ and $\sigma[i]=-1$, and does nothing if $\sig[i]\in\{\used,\unused\}$.

A valid signature must ensure that no step deletes a vertex that is
absent or adds a vertex that is already present, and that the set of vertices
obtained after applying reconfiguration steps to $S_s \cap X$ is the set
$S_t \cap X$. Additionally, because $\Pi$ is hereditary, we can check whether this
property is at least locally satisfied (in $G[X]$) after each step of the sequence.
More formally, we have the following definition.

\begin{definition}\label{def-valid}
A signature $\sig$ over $X$ is {\em valid} if
\begin{enumerate}[label=(\arabic*),topsep=2pt,itemindent=10pt]
\item $\sig[i] \in \sig(i-1,S_s \cap X)$ whenever $\sig[i]\in X$ and $\sigma[i]=-1$,
\item $\sig[i] \not\in \sig(i-1,S_s \cap X)$ whenever $\sig[i]\in X$ and $\sigma[i]=+1$,
\item $\sig(\ell, S_s \cap X) = S_t \cap X$, and
\item $G[X \setminus\sig(i, S_s \cap X)] \in \Pi$ for all $i \leq \ell$.
\end{enumerate}
\end{definition}

It is not hard to see that
a signature $\sig$ over $X$ is valid if and only if
$\sig(0,S_s\cap X),\dots,\sig(\ell,S_s\cap X)$ is a well-defined path
between $S_s \cap X$ and $S_t \cap X$ in
$R_{\textsc{min}}(G[X],n)$.
We will consider only valid signatures.
The dynamic programming algorithms will enumerate exactly the signatures that 
can be extended to valid signatures over $V_i$ in the following sense:

\begin{definition}
A signature $\extsig$ over $V_i$ {\em extends} a signature $\pi$ over $\bag_i$ if it is
obtained by replacing some \textnormal{$\used$} steps with vertex steps from $V_i\setminus \bag_i$
\end{definition}

However, for many problems, the fact that $S$ is a solution for $G[X]$ for
each bag $X$ does not imply that $S$ is a solution for $G$, and checking
this `local' notion of validity will not be enough -- the algorithm will have to maintain additional information.
One such example is the \textsc{OCT-R} problem, which we discuss in Section~\ref{sec-octr-fvsr-tw}.

\subsection{An algorithm for VC-R}\label{sec-vc-tw}
To process nodes of the tree decomposition, we now define ways of
generating signatures from other signatures.
The introduce operation determines all ways that an introduced vertex
can be represented in a signature, replacing
$\unused$ steps in the signature of its child.

\begin{definition}%[The introduce operation]
\label{def-introduce}
Given a signature $\sig$ over $X$ and a vertex $v \not\in X$,
the {\em introduce operation}, $introduce(\sig, v)$ returns the following set of signatures over $X \cup \{v\}$:
for every subset $I$ of indices $i$ for which $\sig[i]=\unused$, consider a
copy $\sig'$ of $\sig$ where for all $i\in I$ we set $\sig'[i]=v$, check if it is valid, and if so, add it to the set.
\end{definition}
In particular $\sig \in introduce(\sig, v)$ and $|introduce(\sig, v)| \leq 2^\ell$.
All signatures obtained through the introduce operation are valid, because of the explicit check.

\begin{definition}%[The forget operation]
\label{def-forget}
Given a signature $\sig$ over $X$ and a vertex $v\in X$,
the {\em forget operation}, returns
a new signature $\sig'=forget(\sig, v)$ over $X \setminus \{v\}$
such that for all $i\leq \ell$, we have $\sig'[i]=\textnormal{\used}$ if $\sig[i]=v$ and $\sig'[i]=\sig[i]$ otherwise.
\end{definition}

Since $\sig'(i, S_s\cap X \setminus \{v\}) = \sig(i,S_s\cap X)\setminus \{v\}$, it is
easy to check that the forget operation preserves validity.

\begin{definition}%[The join operation]
\label{def-join}
Given two signatures $\sig_1$ and $\sig_2$ over $X \subseteq V(G)$,
we say $\sig_1$ and $\sig_2$ are \emph{compatible}
if for all $i\leq \ell$:

\begin{enumerate}[label=(\arabic*),nosep,itemindent=10pt]
	\item $\sig_1[i] = \sig_2[i] = \unused$,
	\item $\sig_1[i] = \sig_2[i] = v$ for some $v\in X$, or
	\item either $\sig_1[i]$ or $\sig_2[i]$ is equal to \textnormal{$\used$} and the other is equal to \textnormal{$\unused$}.
\end{enumerate}

\noindent
For two compatible signatures $\sig_1$ and $\sig_2$,
the {\em join operation} returns
a new signature $\sig'=join(\sig_1, \sig_2)$ over $X$
such that for all $i\leq \ell$ we have, respectively:

\begin{enumerate}[label=(\arabic*),nosep,itemindent=10pt]
	\item $\sig'[i] = \textnormal{\unused}$,
	\item $\sig'[i] = v$, and
	\item $\sig'[i] = \textnormal{\used}$.
\end{enumerate}
\end{definition}

\noindent
Since $\sig'=join(\sig_1,\sig_2)$ is a signature over the same
set as $\sig_1$ and differs from $\sig_1$ only by replacing some $\unused$ steps with $\used$ steps, the join operation preserves validity, that is,
if two compatible signatures $\sig_1$ and $\sig_2$ are valid
then so is $\sig' = join(\sig_1, \sig_2)$.

\medskip
Let us now describe the algorithm.
For each $i \in V(T)$ we assign an initially empty table $\table_i$.
All tables corresponding to internal nodes of $T$ will be updated by simple applications
of the introduce, forget, and join operations. %defined in Section~\ref{sec-sign}.

\smallskip\noindent\textit{Leaf nodes.}
Let $i$ be a leaf node,
that is $\bag_i=\{v\}$ for some vertex $v$.
We let $\table_i = introduce(\sig, v)$, where $\sig$ is the signature with only $\unused$ steps.

\smallskip\noindent\textit{Introduce nodes.}
Let $j$ be the child of an introduce node $i$,
that is $\bag_i = \bag_j \cup \{v\}$ for some $v \not\in \bag_j$.
We let $\table_i =\bigcup_{\sig\in\table_j} introduce(\sig, v)$.

\smallskip\noindent\textit{Forget nodes.}
Let $j$ be the child of a forget node $i$,
that is $\bag_i = \bag_j \setminus \{v\}$ for some $v \in \bag_j$.
We let $\table_i = \{ forget(\sig, v) \mid \sig\in\table_j \}$.

\smallskip\noindent\textit{Join nodes.}
Let $j$ and $h$ be the children of a join node $i$,
that is $\bag_i = \bag_j = \bag_h$.
We let $\table_i = \{join(\sig_j, \sig_h) \mid \sig_j\in\table_j, \sig_h\in\table_h, \mbox{ and }\sig_j\mbox{ is compatible with }\sig_h\}$.

\smallskip
The operations were defined so that the following lemma holds by induction.
The theorem then follows by making the algorithm accept when
$\table_{root}$ contains a signature $\sig$ such that no step of $\sig$ is $\unused$.

\begin{lemma}[*]\label{lem-dp-induction}
For $i \in V(T)$ and a signature $\sig$ over $\bag_i$, $\sig\in \table_i$
if and only if
$\sig$ can be extended to a signature over $V_i$ that is valid.
\end{lemma}

\begin{theorem}[*]\label{th-vcr-tw}
\textsc{VC-R} and \textsc{IS-R} can be solved in
$\Ohstar(4^{\ell} (t+3)^{\ell})$ time on graphs of treewidth $t$.
\end{theorem}

\subsection{VC-R in planar graphs}\label{sec-vcr-shifting}
Using an adaptation of Baker's approach for
decomposing planar graphs~\cite{BAKER94}, also
known as the {\em shifting technique}~\cite{BK08,DH08,E00},
we show a similar result for \textsc{VC-R} and \textsc{IS-R} on planar graphs.
The idea is that at most $\ell$ elements of a solution will be changed,
and thus if we divide the graph into $\ell+1$ parts, one of these parts
will be unchanged throughout the reconfiguration sequence.
The shifting technique allows the definition of the $\ell+1$ parts so that
removing one (and replacing it with simple gadgets to preserve all needed information)
yields a graph of treewidth at most $3\ell-1$.

\begin{theorem}[*]\label{th-vcr-planar}
\textsc{VC-R} and \textsc{IS-R} are fixed-parameter tractable
on planar graphs when parameterized by $\ell$. Moreover,
there exists an algorithm which solves both problems
in $\Ohstar(4^{\ell} (3\ell + 2)^{\ell})$ time.
\end{theorem}

We note that, by a simple application of the result of Demaine et al.~\cite{DHK05},
Theorem~\ref{th-vcr-planar} generalizes to $H$-minor-free graphs and only
the constants of the overall running time of the algorithm are affected.

\subsection{An algorithm for OCT-R}\label{sec-octr-fvsr-tw}
In this section we show how known dynamic programming algorithms
for problems on graphs of bounded treewidth can be adapted to reconfiguration.
The general idea is to maintain a view of the reconfiguration
sequence just as we did for VC-R and in addition check if every
reconfiguration step gives a solution, which can be accomplished
by maintaining (independently for each step) any
information that the original algorithm would maintain.
We present the details for \textsc{OCT-R} (where $\Pi$ is the collection
of all bipartite graphs) as an example.

In a dynamic programming algorithm for \textsc{VC} on graphs of bounded treewidth, it is enough
to maintain information about what the solution's intersection with the bag can be.
This is not the case for \textsc{OCT}.
One algorithm for \textsc{OCT} works in time $\Ohstar(3^t)$ by additionally maintaining
a bipartition of the bag (with the solution deleted)~\cite{FHRV08,FG06}.
That is, at every bag $\bag_i$, we would maintain a list of
assignments $\bag\to\{\used,\ttleft,\ttright\}$ with the property
that there exists a subset $S$ of $V_i$ and a bipartition $L,R$ of
$G[V_i\setminus S]$ such that $X_i \cap S, X_i \cap L$, and $X_i \cap R$
are the $\used$, $\ttleft$, and $\ttright$ vertices, respectively.
A signature for \textsc{OCT-R} will hence additionally store a bipartition
for each step (except for the first and last sets $S_s$ and $S_t$, as we
already assume them to be solutions).

\begin{definition}\label{def-oct-sign}
An {\em OCT-signature} $\sig$ over a set $X \subseteq V(G)$ is a sequence of steps $\sig[1],\dots,\sig[\ell]\in X \cup \{\used,\unused\}$
together with an entry $\sig[i,v] \in \{\ttleft,\ttright\}$ for every
$1 \leq i \leq \ell-1$ and $v \in X\setminus \sig(i,S_s \cap X)$.
\end{definition}

There are at most $(t+3)^\ell 2^{t (\ell-1)}$ different OCT-signatures.
In the definition of validity, we replace the last condition with the following, stronger one:

\begin{enumerate}[label=(\arabic*),topsep=2pt,itemindent=10pt]
\item[(4)] For all $1\leq i \leq \ell-1$, the sets $\{v \mid \sig[i,v]=\ttleft\}$ and $\{v \mid \sig[i,v]=\ttright\}$ give a bipartition of $G[X \setminus\sig(p, S_s \cap X)]$.
\end{enumerate}

In the definition of the join operation, we additionally require
two signatures to have equal $\sig[i,v]$ entries (whenever defined) to be
considered compatible; the operation copies them to the new signature.
In the definition of the forget operation, we delete any $\sig[i,v]$
entries, where $v$ is the vertex being forgotten.
In the introduce operation, we consider (and check the validity of) a
different copy for each way of replacing $\unused$ steps with $v$ steps and each 
way of assigning $\{\ttleft,\ttright\}$ values to new $\sig[i,v]$
entries, where $v$ is the vertex being introduced.
As before, to each node we assign an initially empty table of
OCT-signatures and fill them bottom-up using these operations.
Lemma~\ref{lem-dp-induction}, with the
new definitions, can then be proved again by induction.

\begin{theorem}[*]\label{th-oct-tw}
\textsc{OCT-R} and \textsc{IBS-R} can be solved in
$\Ohstar(2^{t \ell} 4^{\ell} (t+3)^{\ell})$ time on graphs of treewidth $t$.
\end{theorem}

Similarly, using the classical $\Ohstar(2^{\Oh(t \log t)})$ algorithm
for \textsc{FVS} and \textsc{IF} (which maintains what partition of $X_i$ the
connected components of $V_i$ can produce), we can get the following
running times for reconfiguration variants of these problems.

\begin{theorem}\label{th-fvs-tw}
\textsc{FVS-R} and \textsc{IF-R}  can be solved in
$\Ohstar(t^{\ell t} 4^\ell (t+3)^{\ell})$
%$\Ohstar(2^{\Oh(\ell t \log t)})$
time on graphs of treewidth $t$.
\end{theorem}

\section{Conclusion}\label{sec-con}
We have seen in Section~\ref{sec-octr-fvsr-tw} that, with only minor modifications,
known dynamic programming algorithms for problems on graphs of bounded
treewidth can be adapted to reconfiguration. It is therefore natural to ask
whether the obtained running times can be improved via more sophisticated algorithms which
exploit properties of the underlying problem or whether these running
times are optimal under some complexity assumptions.
Moreover, it would be interesting to investigate
whether the techniques presented for planar graphs can be
extended to other problems or more general classes of sparse graphs.
In particular, the parameterized complexity of ``non-local'' reconfiguration problems such as
\textsc{FVS-R} and \textsc{OCT-R} remains open even for planar graphs.

\bibliographystyle{abbrv}
\bibliography{references}

\begin{thebibliography}{10}

\bibitem{BAKER94}
B.~S. Baker.
\newblock Approximation algorithms for {NP}-complete problems on planar graphs.
\newblock {\em J. ACM}, 41(1):153--180, Jan. 1994.

\bibitem{BF84}
G.~Bauer and F.~Otto.
\newblock Finite complete rewriting systems and the complexity of the word
  problem.
\newblock {\em Acta Informatica}, 21(5):521--540, Dec. 1984.

\bibitem{B98}
H.~L. Bodlaender.
\newblock A partial k-arboretum of graphs with bounded treewidth.
\newblock {\em Theor. Comput. Sci.}, 209(1--2):1--45, Dec. 1998.

\bibitem{BK08}
H.~L. Bodlaender and A.~M. C.~A. Koster.
\newblock Combinatorial optimization on graphs of bounded treewidth.
\newblock {\em Comput. J.}, 51(3):255--269, May 2008.

\bibitem{B12}
P.~Bonsma.
\newblock The complexity of rerouting shortest paths.
\newblock In {\em Proc. of Mathematical Foundations of Computer Science}, pages
  222--233, 2012.

\bibitem{B12Planar}
P.~Bonsma.
\newblock Rerouting shortest paths in planar graphs.
\newblock In {\em {FSTTCS} 2012}, volume~18 of {\em Leibniz International
  Proceedings in Informatics ({LIPIcs)}}, pages 337--349, 2012.

\bibitem{CVJ08}
L.~Cereceda, J.~van~den Heuvel, and M.~Johnson.
\newblock Connectedness of the graph of vertex-colourings.
\newblock {\em Discrete Mathematics}, 308(56):913--919, 2008.

\bibitem{CVJ11}
L.~Cereceda, J.~van~den Heuvel, and M.~Johnson.
\newblock Finding paths between 3-colorings.
\newblock {\em Journal of Graph Theory}, 67(1):69--82, 2011.

\bibitem{Courcelle90}
B.~Courcelle.
\newblock The monadic second-order logic of graphs. {I}. recognizable sets of
  finite graphs.
\newblock {\em Information and Computation}, 85(1):12 -- 75, 1990.

\bibitem{DHK05}
E.~Demaine, M.~Hajiaghayi, and K.~Kawarabayashi.
\newblock Algorithmic graph minor theory: secomposition, approximation, and
  coloring.
\newblock In {\em Proceedings of the 46th Annual IEEE Symposium on Foundations
  of Computer Science}, pages 637--646, Oct 2005.

\bibitem{DH08}
E.~D. Demaine and M.~Hajiaghayi.
\newblock The bidimensionality theory and its algorithmic applications.
\newblock {\em Comput. J.}, 51(3):292--302, 2008.

\bibitem{D05}
R.~Diestel.
\newblock {\em {Graph Theory}}.
\newblock Springer-Verlag, Electronic Edition, 2005.

\bibitem{DF97}
R.~G. Downey and M.~R. Fellows.
\newblock {\em Parameterized complexity}.
\newblock Springer-Verlag, 1997.

\bibitem{E00}
D.~Eppstein.
\newblock Diameter and treewidth in minor-closed graph families.
\newblock {\em Algorithmica}, 27(3):275--291, 2000.

\bibitem{FHRV08}
S.~Fiorini, N.~Hardy, B.~Reed, and A.~Vetta.
\newblock Planar graph bipartization in linear time.
\newblock {\em Discrete Appl. Math.}, 156(7):1175--1180, Apr. 2008.

\bibitem{FG06}
J.~Flum and M.~Grohe.
\newblock {\em Parameterized complexity theory}.
\newblock Springer-Verlag, 2006.

\bibitem{GKMP09}
P.~Gopalan, P.~G. Kolaitis, E.~N. Maneva, and C.~H. Papadimitriou.
\newblock The connectivity of boolean satisfiability: computational and
  structural dichotomies.
\newblock {\em SIAM J. Comput.}, 38(6):2330--2355, 2009.

\bibitem{Grohe07}
M.~Grohe.
\newblock Logic, graphs, and algorithms.
\newblock {\em Electronic Colloquium on Computational Complexity (ECCC)},
  14(091), 2007.

\bibitem{IDHPSUU11}
T.~Ito, E.~D. Demaine, N.~J.~A. Harvey, C.~H. Papadimitriou, M.~Sideri,
  R.~Uehara, and Y.~Uno.
\newblock On the complexity of reconfiguration problems.
\newblock {\em Theor. Comput. Sci.}, 412(12-14):1054--1065, 2011.

\bibitem{IKD12}
T.~Ito, M.~Kami\'{n}ski, and E.~D. Demaine.
\newblock Reconfiguration of list edge-colorings in a graph.
\newblock {\em Discrete Applied Mathematics}, 160(15):2199--2207, 2012.

\bibitem{KMM11}
M.~Kami\'{n}ski, P.~Medvedev, and M.~Milani\v{c}.
\newblock Shortest paths between shortest paths.
\newblock {\em Theor. Comput. Sci.}, 412(39):5205--5210, 2011.

\bibitem{K94}
T.~Kloks.
\newblock {\em Treewidth, Computations and Approximations}, volume 842 of {\em
  Lecture Notes in Computer Science}.
\newblock Springer, 1994.

\bibitem{MNRSS13}
A.~E. Mouawad, N.~Nishimura, V.~Raman, N.~Simjour, and A.~Suzuki.
\newblock On the parameterized complexity of reconfiguration problems.
\newblock In {\em Proc. of the 8th International Symposium on Parameterized and
  Exact Computation}, pages 281--294, 2013.

\bibitem{N06}
R.~Niedermeier.
\newblock {\em Invitation to fixed-parameter algorithms}.
\newblock Oxford University Press, 2006.

\bibitem{P47}
E.~L. Post.
\newblock {Recursive unsolvability of a problem of Thue}.
\newblock {\em The Journal of Symbolic Logic}, 12(1):1--11, 1947.

\bibitem{S70}
W.~J. Savitch.
\newblock Relationships between nondeterministic and deterministic tape
  complexities.
\newblock {\em Journal of computer and system sciences}, 4(2):177--192, 1970.

\bibitem{W14}
M.~Wrochna.
\newblock Reconfiguration in bounded bandwidth and treedepth, 2014.
\newblock arXiv:1405.0847.

\end{thebibliography}

\appendix
\section*{Appendix}
\renewcommand{\thesubsection}{\Alph{subsection}}

\subsection{Details omitted from Section~\ref{sec-hardness}}
\subsubsection{Proof of Lemma~\ref{thue}}
\begin{proof}
We note that Bauer and Otto's explicit proof for
$c$-balanced Thue systems~\cite{BF84} can easily be adapted to
give a 2-balanced Thue system. We include a self-contained proof here for completeness.

Since only words of the same length can be reached by application of
rules in a balanced Thue system, it suffices to nondeterministically search all
words of the same length to solve the problem in
nondeterministic polynomial space.
By Savitch's Theorem \cite{S70}, this places the problem in PSPACE.

Let $M=(\Sigma,Q,q_0,q_{acc},q_{rej},\delta)$ be a deterministic Turing Machine working
in space bounded by a polynomial $p(|x|)$, where $p$ is a polynomial function and $x\in\Sigma^*$, which accepts any {PSPACE}-complete language.
(By starting from a fixed {PSPACE}-complete problem we show the word problem to
be hard for a certain fixed Thue system; starting from any language in PSPACE we would
only show that the more general word problem, where the system is given as input, is {PSPACE}-complete).
$\Sigma$ is the tape alphabet of $M$, $Q$ is the set of states, $q_0,q_{acc},q_{rej}$
are the initial, accepting, and rejecting state respectively,
and $\delta: Q\times \Sigma \to Q\times \Sigma \times \{\cdot,L,R\}$ is the transition function of $M$.
Let $\tapeL,\tapeR\in\Sigma$ denote the left and right end-markers.
We assume without loss of generality that the machine clears the tape and
moves its head to the left end when reaching the accepting state.

For any input $x\in\Sigma^*$ we encode a configuration of the Turing Machine by a
word of length exactly $p(|x|)$ over the alphabet $\Gamma = \Sigma \cup (\Sigma \times Q) \cup\{\tapeE\}$.
If the tape content is $\tapeL a_1 a_2 \dots a_n \tapeR$ for
some $a_1,\dots,a_n\in \Sigma$, the head's position is $i\in\{0,1,\dots,n+1\}$ and
the machine's state is $q$, then we define the corresponding word to be the tape
content padded with $\tapeE$ symbols and with $a_i$ replaced by $(q,a_i)$, that
is $\tapeL{}a_1\dots{}a_{i-1}(q, a_i)a_{i+1}\dots{}a_n\tapeR \tapeE\tapeE \dots \tapeE \in \Gamma^{p(|x|)}$.
The initial configuration is then encoded
as $s_x=(q_0,\tapeL){x}\tapeR\tapeE\tapeE\dots\tapeE\in \Gamma^{p(|x|)}$ and the only possible
accepting configuration is encoded as $t_x=(q_{acc},\tapeL)\tapeR\tapeE\tapeE\dots\tapeE \in \Gamma^{p(|x|)}$.
Since $M$ never uses more than $p(|x|)$ space on input $x$, our encoding is well
defined for all configurations appearing in the execution of $M$ on $x$.
So $M$ accepts input $x$ if and only if from $s_x$ one reaches the
configuration $t_x$ by repeatedly applying the transition function.
Such an application corresponds exactly to
the following (ordered) string rewriting rules, in the encodings:
\begin{itemize}
\item $\big( (q,a)c\ ,\ (p,b)c \big)$\quad for $q\in Q$,\ $a,c\in\Sigma$ and $\delta(q,a)=(p,b,\cdot)$,
\item $\big( (q,a)c\ ,\ b(p,c) \big)$\quad for $q\in Q$,\ $a,c\in\Sigma$ and $\delta(q,a)=(p,b,R)$,
\item $\big( c(q,a)\ ,\ (p,c)b \big)$\quad for $q\in Q$,\ $a,c\in\Sigma$ and $\delta(q,a)=(p,b,L)$.
\end{itemize}

The transition relation is not symmetric, but since the machine
$M$ is deterministic, the configuration digraph (with machine configurations
as vertices and the transition function as the adjacency relation) has out-degree one.
The configuration $t_x$ (which is a configuration in the accepting state)
has a loop, i.e. a directed edge from $t_x$ to $t_x$.
Therefore from any configuration, $t_x$ is reachable by a directed path
if and only if it is reachable by any path. This means that $M$ accepts
input $x$ if and only if applying the transition rules to $s_x$ leads to
$t_x$ if and only if $s_x \reach_R t_x$,  where $R$ is the symmetric
closure of the above rules, i.e., the 2-balanced Thue system over $\Gamma$ with rules:
\begin{itemize}
\item $\{(q,a)c,(p,b)c\}$\quad for $q\in Q$,\ $a,c\in\Sigma$ and $\delta(q,a)=(p,b,\cdot)$,
\item $\{(q,a)c , b(p,c)\}$\quad for $q\in Q$,\ $a,c\in\Sigma$ and $\delta(q,a)=(p,b,R)$,
\item $\{c(q,a) , (p,c)b\}$\quad for $q\in Q$,\ $a,c\in\Sigma$ and $\delta(q,a)=(p,b,L)$.
\end{itemize}
Since the map $x\mapsto (s_x,t_x)$ is computable in logarithmic space, this proves
the word problem of $(\Gamma,R)$ to be {PSPACE}-hard.
\qed
\end{proof}

\subsubsection{Proof of Lemma~\ref{hWordReconfiguration}}
\begin{proof}
We first need to slightly strengthen Lemma~\ref{thue} to give a Thue system where
only one symbol at a time can be changed. To that aim, it suffices to replace a
rule changing two symbols with a sequence of rules using two new intermediary symbols.
\begin{clm}\label{smallthue}
There is a 2-balanced Thue system $(\Gamma,R)$ whose word problem is
PSPACE-complete and such that for every rule $\{a_1a_2,b_1b_2\}\in R$ either $a_1=b_1$ or $a_2=b_2$.
\end{clm}
\begin{proof}
Let $(\Sigma,R)$ be the 2-balanced Thue system from Lemma~\ref{thue}.
Suppose $\{a_1a_2,b_1b_2\}$ is a rule of $R$ in which $a_1\neq b_1$ and
$a_2\neq b_2$. We construct a 2-balanced Thue system $(\Gamma,S)$ with one
fewer such rule, preserving {PSPACE}-completeness of the word problem.
The claim then follows inductively.

Let $\Gamma=\Sigma\cup\{X,Y\}$, where $X$ and $Y$ are new symbols which will be
used to replace a rule changing two symbols with a sequence of rules changing only one symbol.
Let $S = R \setminus \{\{a_1a_2,b_1b_2\}\} \cup \{\{a_1a_2,Xa_2\},\{X a_2, XY\}, \{XY,b_1 Y\}, \{b_1 Y, b_1 b_2\}\}$.
We show that for any $s,t\in\Sigma^*$ it holds that $s\reach_R t$ if and only if
$s\reach_S t$, which implies that our construction preserves {PSPACE}-completeness.

Clearly if $s\reach_R t$ then $s\reach_S t$, because replacing $a_1a_2$
with $b_1b_2$ can be done in $S$ by replacing $a_1a_2$ with $Xa_2$,
then $XY$, then $b_1 Y$ and finally $b_1b_2$. Suppose now $s\reach_S t$ for
some $s,t\in\Sigma^*$. Then there is a sequence $s=u_0,u_1,u_2,\dots,u_l=t$ of
words $u_i\in\Gamma^*$ such that $u_i \ruler_S u_{i+1}$.
Let $\phi : \Gamma^*\to \Sigma^*$ be defined by replacing all $XY$
substrings of a word with $a_1a_2$, then replacing all remaining $X$
symbols with $a_1$ and all remaining $Y$ symbols with $b_2$. It is easy
to check that $\phi(u_i)\ruler_R\phi(u_{i+1})$ or $\phi(u_i)=\phi(u_{i+1})$.
Since $\phi(u_0)=\phi(s)=s$ and $\phi(u_l)=\phi(t)=t$, this implies that $s\reach_R t$.
\qed
\end{proof}

Let $(\Gamma,S)$ be the 2-balanced Thue system from
Lemma~\ref{smallthue} (so if $\{a_1a_2,b_1b_2\}\in S$ then
$a_1=b_1$ or $a_2=b_2$). Let $S=\{S_1,\dots,S_m\}$.

Let $\tapeE,\tapeL,\tapeR,x_1,\dots,x_m$ be new symbols, let
$\Delta_1 = \{\tapeL,\tapeR,x_1,\dots,x_m\}$,
$\Delta_2 = (\Gamma\cup\{\tapeE\})\times(\Gamma\cup\{\tapeE\})$, and
let $\Delta = \Delta_1 \cup \Delta_2$. We will call $\Delta_1$ \emph{special symbols}
and $\Delta_2$ \emph{pair symbols}. Let $H=(\Delta,E$), where we
define $E\subseteq \Delta^2$ as the relation containing the following pairs
\begin{itemize}
\item $((a,b),(b,c))$\quad for any $a,b,c\in \Gamma$,
\item $(\tapeL, (\tapeE,a))$\quad for any $a\in \Gamma$,
\item $((a,\tapeE),\tapeR)$\quad for any $a,b,c\in \Gamma$,
\item $((\cdot,a_1),x_i)$,
\item $((\cdot,b_1),x_i)$,
\item $(x_i, (a_2,\cdot))$,
\item $(x_i, (b_2,\cdot))$ \quad for  any $\cdot\in\Gamma$ and $i\in\{1,\dots,m\}$ such that $S_i=\{a_1a_2,b_1b_2\}$.
\end{itemize}

Let $(s,t)\in \Gamma^*\times \Gamma^*$ be an instance of the word problem for
$S$, without loss of generality $|s|=|t|=n$. Define $\psi:\Gamma^n\to \Delta^{n+3}$ as
$$\psi(a_1a_2\dots a_n) = \tapeL (\tapeE,a_1)(a_1,a_2)(a_2,a_3)\dots(a_{n-1},a_n)(a_n,\tapeE) \tapeR$$
It is easy to see that if $s\reach_S t$ then $\psi(s)$ can be
transformed into $\psi(t)$, e.g., applying the rule $S_i=\{a_1 a_2, b_1 a_2\}$
corresponds to replacing $(\cdot, a_1)(a_1,a_2)(a_2,\cdot)$ by $(\cdot, a_1)x_i(a_2,\cdot)$, then
$(\cdot, b_1)x_i(a_2,\cdot)$, then $(\cdot, b_1)(b_1,a_2)(a_2,\cdot)$. We will
show the other direction, that if $\psi(s)$ can be transformed into $\psi(t)$, then
$s\reach_S t$. Since $\psi$ is computable in logarithmic
space, this will imply our claim of {PSPACE}-completeness.

Indeed, suppose that there is a sequence of $H$-words $\psi(s)=u_0,u_1,\dots,u_l=\psi(t)$
with $u_j\in\Delta^{n+3}$, such that $u_j$ differs from $u_{j+1}$ only at one position.
In any $H$-word $v=v_1 v_2\dots v_{n+3}\in\Delta^{n+3}$ there cannot be two
consecutive special symbols. We can thus define a word $\phi(v)$ of length $n$
over $\Gamma$ such that its $i$th symbol, for $i\in\{1,\dots,n\}$, is the second
element of $v_{i+1}$ if $v_{i+1}$ is a pair symbol and the first
element of $v_{i+2}$ if $v_{i+2}$ is a pair symbol (either case must hold
and if both do, the definitions agree by construction of $E$).
In particular $\phi(\psi(v))=v$ for any $v\in\Gamma^n$. We argue
that $\phi(u_{j-1})\reach_S \phi(u_{j})$ for $j\in\{1,\dots,l\}$.

Notice that the special symbol $\tapeL$ must precede a pair symbol
$(\tapeE,\cdot)$ for some $\cdot\in \Gamma$ and any such pair symbol
must be preceded by $\tapeL$. Since only one symbol at a time can
be changed, it follow inductively that for each $j\in\{0,\dots,l\}$ the first
two symbols of $u_j$ must be $\tapeL (\tapeE,\cdot)$ for
some $\cdot\in\Gamma$ and $\tapeL$ appears nowhere else. A similar argument
applies to the last two symbols, $(\cdot,\tapeE) \tapeR$ for some $\cdot\in\Gamma$.

Since $u_{j-1}$ and $u_j$ differ at only one position, there are
non-empty words $v,w\in\Delta^*$ and symbols $a,b\in\Delta$, $a\neq b$ such
that $u_{j-1}=vaw$ and $u_j=vbw$. If $a$ or $b$ is a special symbol then both
the last symbol of $v$ and the first symbol of $w$ are pair symbols, so
$\phi(u_{j-1})=\phi(u_{j})$. Otherwise, let $a=(a_1,a_2), b=(b_1,b_2)$. Assume
without loss of generality that $a_1\neq b_1$ and $a_2=b_2$
(the case $a_1=b_1,a_2\neq b_2$ is analogous and the case $a_1\neq b_1,a_2\neq b_2$
can be split by showing that $\phi(u_{j-1})\reach_S \phi(u')$ and
$\phi(u') \reach_S \phi(u_j)$ for $u'=v(b_1,a_2)w$, which can easily be checked to be an $H$-word).
If the last symbol of $v$ is a pair symbol $(c,d)$, then $d=a_1$ and $d=b_1$, contradicting
our assumption. If the last symbol of $v$ is $\tapeL$, then $a_1=b_1=\tapeE$. Finally
if the last symbol of $v$ is $x_i$ for some $i\in\{1,\dots,m\}$, then $S_i$
must equal $\{ca_1,c'b_1\}$ for some $c,c'\in\Gamma$. Since $a_1\neq b_1$, we
have $c=c'$ and the second-to-last symbol of $v$ must be a pair $(\cdot,c)$
for some $\cdot\in\Gamma\cup\{\tapeE\}$. Thus $\phi(v(b_1,a_2)w)$ is obtained
from $\phi(v(a_1,a_2)w)$ by replacing the symbol $a_1$ at position $|v|$, which is
preceded by a $c$, by the symbol $b_1$, that is, $\phi(v(b_1,a_2)w) \ruler_S \phi(v(a_1,a_2)w)$.
\qed
\end{proof}

\subsection{Details omitted from Section~\ref{sec-vc-tw}}
\subsubsection{Proof of Lemma~\ref{lem-dp-induction}}
\begin{proof}
We first prove a few statements about signature validity.
Note that all signatures in the algorithm are obtained through join, forget or introduce operations, which preserve validity and thus for each $i \in V(T)$, the table $\table_i$ contains only valid signatures over $X_i$.

\begin{lemma}\label{lem-subsig-valid}
If a signature $\sig$ over $X$ is obtained from a valid signature by replacing all vertex steps not in $X$ by $\used$ or $\unused$ steps,
then $\sig$ is valid as well.
\end{lemma}
\begin{proof}
Let $\sig$ be obtained from a valid signature $\sig'$ over $X'$ by replacing all vertex steps in $X'\setminus X$ by $\used$ or $\unused$ steps.
First note that $\sig(i,S_s \cap X)=\sig'(i,S_s \cap X')\cap X$.
The first three conditions of Definition~\ref{def-valid} follow immediately.
As $\Pi$ is hereditary, $G[X'\setminus S]\in \Pi$ implies $G[X \setminus (S\cap X)]\in \Pi$, hence the fourth condition also follows.
\qed
\end{proof}

\begin{lemma}\label{lem-vc-cover}
Let $G$ be a graph $S,X_1,X_2$ be subsets of $V(G)$ such that every edge of $G[X_1 \cup X_2]$ is contained in $G[X_1]$ or $G[X_2]$.
If $S\cap X_1$ is a vertex cover of $G[X_1]$, $S\cap X_2$ is a vertex cover of $G[X_2]$ then $S$ is a vertex cover of $G[X_1 \cup X_2]$.
\end{lemma}
\begin{proof}
Let $uv$ be an edge of $G[X_1 \cup X_2]$.
Then it is an edge of $G[X_i]$ for some $i\in\{1,2\}$.
Hence it one of $u,v$ must be a member of $S\cap X_i$.
Thus ever edge of $G[X_1 \cup X_2]$ has an endpoint in $S$.
\qed
\end{proof}

\begin{corollary}\label{cor-vc-cover-valid}
Let $\sig,\sig_1,\sig_2$ be a signatures over $X,X_1,X_2$ respectively, such that $X=X_1\cup X_2$ and every edge of $G[X]$ is contained in $G[X_1]$ or $G[X_2]$.
Assume furthermore that for all $i\leq\ell$:
\begin{itemize}[nosep]
\item[] $\sig[i]=\sig_1[i]$ whenever $\sig[i]\in X_1$ or $\sig_1[i]\in X_1$ and
\item[] $\sig[i]=\sig_2[i]$ whenever $\sig[i]\in X_2$ or $\sig_2[i]\in X_2$.
\end{itemize}
If $\sig_1$ and $\sig_2$ are valid, then so is $\sig$.
\end{corollary}
\begin{proof}
The assumption means that $\tau$ and $\tau_1$ agree over all changes within $X_1$, that is, $\tau_1(i,S_s \cap X_1) = \tau(i,S_s \cap X) \cap X_1$ (and similarly for $\tau_2$).
The first two conditions of Definition~\ref{def-valid} for $\sig$ follow immediately: if $\sig[i]\in X$ then $\sig[i]\in X_1$ or $\sig[i]\in X_2$,
so the statement is equivalent to the first two conditions for $\tau_1$ or for $\tau_2$.
To show the third condition for $\tau$, observe that 
$\tau(i,S_s \cap X)  = \left(\tau(i,S_s \cap X) \cap X_1\right) \cup \left(\tau(i,S_s \cap X) \cap X_2\right) = \tau_1(i,S_s \cap X_1) \cup \tau_2(i,S_s \cap X_2) = \left(S_t \cap X_1\right) \cup \left(S_t \cap X_2\right) = S_t \cap X$.
For the last condition, it suffices to use Lemma~\ref{lem-vc-cover} for $S=\tau(i,S_s \cap X)$.
\qed
\end{proof}

We now prove Lemma~\ref{lem-dp-induction}:
\emph{For $i \in V(T)$ and a signature $\sig$ over $\bag_i$, $\sig\in \table_i$
if and only if
$\sig$ can be extended to a signature over $V_i$ that is valid.}
We prove the statement by induction over the tree $T$, that is,
we prove the statement to be true at $i \in V(T)$ assuming we have already proved it for all other nodes in the subtree of $T$ rooted at $i$.
Depending on whether $i$ is a leaf, forget, introduce or join node,
we have the following cases. 
\smallskip

\noindent\textit{Leaf nodes.}
Let $v$ be the only vertex of $\bag_i$, that is, $V_i = \bag_i = \{v\}$.
Since $V_i=\bag_i$, a signature $\sig$ over $\bag_i$ can be extended to a signature valid over $V_i$ if and only if $\sig$ is valid and has no $\used$ steps.
That is, if and only if $\sig$ has only $\unused$ and $v$ steps and is valid (over $\bag_i$),
which happens if and only if $\sig \in \table_i$.
\smallskip

\noindent\textit{Forget nodes.}
Let $j$ be the child of $i$, thus
$\bag_i = \bag_j \setminus \{v\}$ for some $v \in \bag_j$ and $V_i=V_j$.

For one direction, suppose $\sig\in \table_i$ over $\bag_i$.
Then there is a $\sig_j$ in $\table_j$ over $\bag_j$ such that $\sig=forget(\sig_j,v)$.
By inductive assumption, $\sig_j$ has an extension $\extsig$ valid over $V_j=V_i$.
Since $\sig_j$ is be obtained from $\sig$ by replacing some $\used$ steps with $v$ steps, $\extsig$ is also an extension of $\sig$.
Thus $\sig$ has an extension valid over $V_i$.

For the other direction, suppose $\sig$ has an extension $\extsig$ valid over $V_i$.
Then $\extsig$ is obtained from $\sig$ by replacing some $\used$ steps with vertex steps from $V_i\setminus \bag_i$.
Since $V_i\setminus \bag_i = (V_j\setminus \bag_j) \cup \{v\}$,
we can consider the signature $\sig_j$ over $\bag_j \cup \{v\}$ obtained by only using the replacements with $v$ steps.
This signature $\sig_j$ can be extended to $\extsig$ by using the remaining replacements, so by inductive assumption $\sig_j\in \table_j$.
Furthermore, $forget(\sig_j,v) = \sig$.
Thus $\sig\in \table_i$.

\begin{center}
	\begin{tabular}{c|c@{\hskip 12pt}c@{\hskip 12pt}c@{\hskip 12pt}c}
	$\sig$   & \unused & \used & \used & $\bag_j$ \\
	$\sig_j$ & \unused & \used & $v$ & $\bag_j$\\
	$\extsig$  & \unused & $V_i\setminus \bag_i$ & $v$ & $\bag_j$\\
	\end{tabular}
\end{center}

\noindent\textit{Introduce nodes.}
Let $j$ be the child of $i$, thus
$\bag_i = \bag_j \cup \{v\}$ for some $v \in \bag_i$ and $V_i=V_j\cup\{v\}$.

For one direction, suppose $\sig\in \table_i$ is a signature over $\bag_i$.
Then there is a $\sig_j\in \table_j$ such that $\sig$ can be obtained from $\sig_j$ by replacing some $\unused$ steps with $v$ steps.
By inductive assumption $\sig_j$ has a extension $\extsig_j$ over $V_j$ that is valid.
As $\extsig_j$ can be obtained from $\sig_j$ by replacing $\used$ steps with vertex steps from $V_j\setminus \bag_j$ and $\sig$ has $\used$ steps at the same positions, we can use the same replacements to obtain an extension $\extsig$ over $V_j\cup\{v\}$ of $\sig$.
$\extsig$ agrees with $\extsig_j$ over $V_j$ and with $\sig$ over $\bag_i$,
it is thus valid over $V_i$ by Corollary~\ref{cor-vc-cover-valid}.
Therefore $\sig$ has an extension over $V_i$ that is valid.

For the other direction, suppose $\sig$ has an extension $\extsig$ over $V_i$ that is valid.
Let $\extsig_j$ be the signature over $V_j=V_i\setminus\{v\}$ obtained by replacing all $v$ steps of $\extsig$ with $\unused$ steps.
By Lemma~\ref{lem-subsig-valid}, $\extsig_j$ is valid.
Let $\sig_j$ be the signature over $\bag_j=\bag_i\setminus\{v\}$ obtained by replacing all $v$ steps of $\sig$ with $\unused$ steps.
Then $\extsig_j$ is an extension of $\sig_j$, thus $\sig_j\in \table_j$ by inductive assumption.
Since $\extsig$ is valid, so is $\sig$ (Lemma~\ref{lem-subsig-valid}), thus $\sig\in introduce(\sig_j, v)$ and $\sig \in \table_i$.

\begin{center}
	\begin{tabular}{c|c@{\hskip 12pt}c@{\hskip 12pt}c@{\hskip 12pt}c}
	$\sig$   & \unused & \used & $v$ & $\bag_j$ \\
	$\sig_j$ & \unused & \used & \unused & $\bag_j$\\
	$\extsig_j$  & \unused & $V_j\setminus \bag_j$ & \unused & $\bag_j$\\
	$\extsig$    & \unused & $V_j\setminus \bag_j$ & $v$ & $\bag_j$ \\
	\end{tabular}
\end{center}
\medskip

\noindent\textit{Join nodes.}
Let $j,h$ be the children of $i$, thus $V_i=V_j\cup V_h$ and we will write $\bag$ for $\bag_i=\bag_j=\bag_h$.

For one direction suppose $\sig\in \table_i$ valid over $\bag$.
Then there are two compatible signatures $\sig_j \in \table_j, \sig_h \in \table_h$ such that $\sig=join(\sig_j,\sig_h)$.
By inductive assumption, they have valid extensions, $\extsig_j$ over $V_j$ and $\extsig_h$ over $V_h$, respectively.
Let $I_i,I_j,I_h$ be the sets of indices of $\used$ steps in $\sig,\sig_j,\sig_h$, respectively.
By Definition~\ref{def-join}, $I_i$ is the sum of disjoint sets $I_j, I_h$.
Since $\extsig_j$ is obtained from $\sig_j$ by replacing steps at indices $I_j$ with vertex steps from $V_j\setminus \bag$ and similarly for $\extsig_h$, we can define a signature $\extsig$ obtained from $\sig$ over $\bag \cup (V_j\setminus \bag) \cup (V_h\setminus \bag) = V_i$ by using both sets of replacements.
$\extsig$ is an extension of $\sig$.
Moreover, $\extsig$ agrees with $\extsig_j$ over $V_j$ and with $\extsig_h$ over $V_h$, so by Corollary~\ref{cor-vc-cover-valid}, $\extsig$ is valid over $V_j \cup V_h = V_i$.
Therefore $\sig$ has a extension over $V_i$ that is valid.

For the other direction, suppose $\sig$ has an extension $\extsig$ over $V_i$ that is valid.
Let $I_j$ be the set of indices of vertex steps from $V_j\setminus \bag$ in $\extsig$ and define $I_h$ accordingly.
Let $\sig_j,\extsig_j$ be obtained from $\sig,\extsig$ by replacing all steps at indices $I_h$ by $\unused$ steps.
Since $V_j\cap V_h = \bag$, $\extsig_j$ is an extension of $\sig_j$ over $V_j$.
By Lemma~\ref{lem-subsig-valid} $\extsig_j$ is valid, thus by inductive assumption $\sig_j \in \table_j$.
Define $\sig_h,\extsig_h$ accordingly and observe that $\sig_h\in \table_h$.
It is easy to see that $\sig$ has $\used$ steps exactly at the indices $I_j \cup I_h$ and $\sig_j,\sig_h$ have $\used$ steps exactly at the disjoint sets of indices $I_j,I_h$, respectively. This implies $\sig_j,\sig_h$ are compatible and $\sig = join(\sig_j,\sig_h)$,
so $\sig \in A_i$.

\begin{center}
	\begin{tabular}{c|c@{\hskip 12pt}c@{\hskip 12pt}c@{\hskip 12pt}c}
	         &         & $I_j$ & $I_h$ & \\\hline
	$\sig$   & \unused & \used & \used & $\bag$ \\
	$\sig_j$ & \unused & \used & \unused & $\bag$\\
	$\sig_h$ & \unused & \unused & \used & $\bag$\\
	$\extsig_j$  & \unused & $V_j\setminus \bag$ & \unused & $\bag$\\
	$\extsig_h$  & \unused & \unused & $V_h\setminus \bag$ & $\bag$\\
	$\extsig$    & \unused & $V_j\setminus \bag$ & $V_h\setminus \bag$ & $\bag$ \\
	\end{tabular}
\end{center}
\qed
\end{proof}

\subsubsection{Proof of Theorem~\ref{th-vcr-tw}}
\begin{proof}
Recall that we say $G \models \omega_\sigma(S_s,S_t)$
if there is a reconfiguration sequence (of vertex covers of $G$) of length exactly $\ell$ from $S_s$ to $S_t$, such that the $i^{th}$ step is a vertex removal if $\sigma[i]=-1$ and a vertex addition if $\sigma[i]=+1$.
The following lemma states the correctness of the acceptance condition of our algorithm.

\begin{lemma}\label{lem-dp-correct}
$G \models \omega_\sigma(S_s,S_t)$
%There exists a reconfiguration sequence of length exactly $\ell$ from $S_s$ to $S_t$
if and only if $\table_{root}$ contains a signature $\sig$ over $\bag_{root}$
such that no step of $\sig$ is $\unused$.
\end{lemma}
\begin{proof}
From Lemma~\ref{lem-dp-induction}, we know that 
$\table_{root}$ contains a signature $\sig$ over $\bag_{root}$ such that no step of $\sig$ is $\unused$
if and only if there is a signature $\extsig$ over $V_{root}=V$ that is valid and such that no step of $\extsig$ is $\unused$.
This means that  $\extsig$ contains only vertex steps and by definition of validity, the corresponding
sequence $\extsig(0,S_s),\dots,\extsig(\ell,S_s)$ is a reconfiguration sequence 
of length exactly $\ell$ from
$S_s$ to $S_t$ such that the $i^{th}$ step is a vertex removal if $\sigma[i]=-1$ and a vertex addition if $\sigma[i]=+1$.
\qed
\end{proof}

It remains to prove the bound on the running time of our algorithm.
The number of nodes in $T$ is in $\Oh(n)$. Checking the compatibility
and validity of signatures can be accomplished in time polynomial in $\ell,t,n$.
For each node $i \in V(T)$ the table $\table_i$ contains at most $(t+3)^\ell$ signatures.
Updating tables at the leaf nodes requires $\Ohstar(2^\ell)$ time, since we
check the validity of $2^\ell$ signatures obtained from one introduce operation.
In the worst case, updating the table of an introduce node requires
$\Ohstar(2^\ell (t+3)^\ell)$ time, i.e. applying the introduce operation on each signature in
a table of size $(t+3)^\ell$. For forget nodes, the time spent
is polynomial in the maximum size of a table, that is $\Ohstar((t+3)^\ell)$.
Finally, updating the table of a join node can be implemented in 
$\Ohstar(2^\ell (t+3)^{\ell})$ time by checking for each of the $(t+3)^\ell$ possible signatures all possible ways to split $\used$ steps among the two children.
The algorithm needs to be run for every $\sigma \in \{-1,+1\}^\ell$ that doesn't violate the maximum allowed capacity, giving in total the claimed  $\Ohstar(4^\ell (t+3)^\ell)$ time bound.

Given an instance $(G, S_s, S_t, k, \ell)$ of \textsc{IS-R},
we can solve the corresponding \textsc{VC-R} instance
$(G, V(G) \setminus S_s, V(G) \setminus S_t, n - k, \ell)$ in
$\Ohstar(4^\ell (t+3)^\ell)$ time on graphs
of treewidth $t$. Combining this fact with Proposition~\ref{prop-dual} yields
the result for \textsc{IS-R}.
\qed
\end{proof}

\subsection{Details omitted from Section~\ref{sec-vcr-shifting}}
\subsubsection{Proof of Theorem~\ref{th-vcr-planar}}
\begin{proof}
Given a plane embedding of a planar graph $G$,
the vertices of $G$ are divided into layers $\{L_1, \ldots , L_r\}$
as follows: Vertices incident to the exterior face are in layer $L_0$.
For $i \geq 0$, we let $G'$ be the graph obtained by deleting
all vertices in $L_0 \cup \ldots \cup L_i$ from $G$. All the vertices that
are incident to the exterior face in $G'$ are in layer $L_{i + 1}$ in $G$.
$L_r$ is thus the last non-empty layer. A planar graph
that has an embedding where the vertices are in $r$ layers is
called {\em $r$-outerplanar}. The following result is due to Bodlaender~\cite{B98}.

\begin{lemma}[Bodlaender~\cite{B98}]\label{lem-planar-tw}
The treewidth of an $r$-outerplanar graph $G$ is at most $3r - 1$.
Moreover, a tree decomposition of width at
most $3r - 1$ can be constructed in time polynomial in $|V(G)|$.
\end{lemma}

From Lemma~\ref{lem-planar-tw}, we have the following corollary.

\begin{corollary}[\cite{B98,BK08}]\label{cor-planar-tw}
For a planar graph $G$, we let $\mathcal{E}$ be an arbitrary
plane embedding of $G$ and $\{L_1, \ldots , L_r\}$ be the
collection of layers corresponding to $\mathcal{E}$.
Then for any $i, \ell \geq 1$, the treewidth of the subgraph
$G[L_{i+1} \cup \ldots \cup L_{i+\ell}]$ is at most $3\ell - 1$.
\end{corollary}

We now summarize the main ideas behind how we use the shifting technique.
Note that every vertex in $S_s \Delta S_t$ must be
touched at least once in any reconfiguration sequence $\alpha$
from $S_s$ to $S_t$. In other words, $S_s \Delta S_t \subseteq V(\alpha)$. Moreover,
we know that $|V(\alpha)|$ is at most $\ell$, as otherwise the corresponding \textsc{VC-R}
instance is a no-instance. For an arbitrary plane embedding of a planar graph $G$
and every fixed $j \in \{0, \ldots, \ell\}$,
we let $G_j$ be the graph obtained by deleting all vertices
in $L_{i(\ell + 1)+j}$, for all $i \in \{0, 1, \ldots, \left \lfloor{{n / (\ell + 1)}}\right \rfloor\}$.
Note that $tw(G_j) \leq 3\ell - 1$.

\begin{proposition}\label{prop-planar-subgraph}
If there exists a reconfiguration sequence $\alpha$ of length exactly $\ell$
between two vertex covers $S_s$ and $S_t$ of a planar graph $G$,
then for some fixed $j \in \{0, \ldots, \ell\}$
we have $V(\alpha) \subseteq V(G_j)$.
\end{proposition}

We still need a few gadgets
before we can apply Theorem~\ref{th-vcr-tw} on each graph $G_j$ and guarantee correctness.
In particular, we need to handle deleted vertices and ``border'' vertices correctly, i.e. vertices
incident to the exterior face in $G_j$.

We solve at most $\left \lfloor{{n / (\ell + 1)}}\right \rfloor + 1$
instances of the \textsc{VC-R} problem as follows:

\begin{itemize}
\item[1.] Find an arbitrary plane embedding of $G$.
\item[2.] For every fixed $j \in \{0, \dots, \ell\}$:
\item[3.] \hspace*{5mm} Let $G^*_j = G_j$.
\item[4.] \hspace*{5mm} Let $D^*_j$ denote the set of vertices deleted from $G$ to obtain $G^*_j$.
\item[5.] \hspace*{5mm} If $\{S_s \Delta S_t\} \cap D^*_j \neq \emptyset$:
\item[6.] \hspace*{10mm} Ignore this instance (continue from line 2).
\item[7.] \hspace*{5mm} Partition $D^*_j$ into $A^*_j = D^*_j \cap \{S_s \cap S_t\}$ and $B^*_j = D^*_j \setminus A^*_j$.
\item[8.] \hspace*{5mm} Let $S^*_{s,j} = S_s \cap V(G^*_j)$ and $S^*_{t,j} = S_t \cap V(G^*_j)$.
\item[9.] \hspace*{5mm} If $\{v \in S^*_{s,j} \Delta S^*_{t,j} \mid |N_G(v) \cap B^*_j| > 0\} \neq \emptyset$:
\item[10.] \hspace*{10mm} Ignore this instance (continue from line 2).
\item[11.] \hspace*{5mm} For every vertex $v \in A^*_j$:
\item[12.] \hspace*{10mm} Add an $(\ell + 1)$-star centered at $u$ to $G^*_j$.
\item[13.] \hspace*{10mm} Add $u$ to $S^*_{s,j}$ and $S^*_{t,j}$.
\item[14.] \hspace*{5mm} For every vertex in $\{v \in \{S^*_{s,j} \cap S^*_{t,j}\} \mid |N_G(v) \cap B^*_j| > 0\}$:
\item[15.] \hspace*{10mm} Add $\ell + 1$ degree-one neighbors to $v$ in $G^*_j$
\item[16.] \hspace*{5mm} Solve instance $(G^*_j, S^*_{s,j}, S^*_{t,j}, k, \ell)$.
\end{itemize}

On lines 5 and 6, we make sure that no vertices from the symmetric difference
of $S_s$ and $S_t$ lie in the deleted layers of $G$, as otherwise $G^*_j$
can be ignored, by Proposition~\ref{prop-planar-subgraph}.
Hence, we know that $D^*_j$ can only include
vertices common to both $S_s$ and $S_t$ (vertices in $S_s \cap S_t$) and we can
partition $D^*_j$ into two sets accordingly (line 7).
In the remaining steps, we add gadgets to account for the capacity
used by vertices in $A^*_j$ and the fact that the neighbors of any
vertex in $B^*_j$ must remain untouched. In other words, we
assume that there exists a reconfiguration
sequence $\alpha$ from $S^*_{s,j}$ to $S^*_{t,j}$ in $R_{\textsc{min}}(G^*_j,k)$.
Then $\alpha$ is a reconfiguration sequence from $S_s$ to $S_t$ in
$R_{\textsc{min}}(G,k)$ only if:\\
\hspace*{2em}(1) $|S^*_{s,j}| + capacity(\alpha) \leq k - |A^*_j|$,\\
where $capacity(\alpha) = \max_{1 \leq \ell' \leq \ell}(\sum_{i = 1}^{\ell'}{sign(\alpha, i)})$ and $sign(\alpha, i)$ is -1 when the $i^{th}$ step of $\alpha$ is a deletion, +1 when it is an addition; and\\
\hspace*{2em}(2) no vertex deletion in $\alpha$ leaves an edge uncovered in $G$.\\
To guarantee property (1), we add an $(\ell + 1)$-star to $G^*_j$ for every vertex in $A^*_j$
then add the center of the star into both $S^*_{s,j}$ and $S^*_{t,j}$ (lines 11, 12, and 13).
Therefore, for every value of $j$ we have $|S| = |S^*_{s,j}|$,
$|T| = |S^*_{t,j}|$, and $|S^*_{s,j}| + capacity(\alpha) \leq k - |A^*_j|$.
For property (2), we add $\ell + 1$ degree-one neighbors to
every vertex in $\{v \in \{S^*_{s,j} \cap S^*_{t,j}\} \mid |N_G(v) \cap B^*_j| > 0\}$ (lines 14 and 15).
Those vertices, as well as the centers of the stars, will have to remain untouched in $\alpha$, as otherwise
deleting any such vertex would require more than $\ell$ additions.

Since adding degree-one vertices and $(\ell + 1)$-stars to a graph does not
increase its treewidth, we have $tw(G^*_j) \leq 3\ell - 1$ for all $j$ (Corollary~\ref{cor-planar-tw}).
Hence, for each graph $G^*_j$ we can now apply Theorem~\ref{th-vcr-tw} and solve the
\textsc{VC-R} instance $(G^*_j, S^*_{s,j}, S^*_{t,j}, k, \ell)$
in $\Ohstar(4^{\ell} (3\ell + 1)^{\ell})$ time.
We prove in Lemma~\ref{lem-vcr-planar} that our original instance on planar $G$
is a yes-instance if and only if $(G^*_j, S^*_{s,j}, S^*_{t,j}, k, \ell)$
is a yes-instance for some fixed $j \in \{0, 1, \ldots, \left \lfloor{{n / (\ell + 1)}}\right \rfloor\}$.

\begin{lemma}\label{lem-vcr-planar}
$(G, S_s, S_t, k, \ell)$ is a yes-instance of \textsc{VC-R} if and only if
$(G^*_j, S^*_{s,j}, S^*_{t,j}, k, \ell)$ is a yes-instance for some fixed
$j \in \{0, 1, \ldots, \left \lfloor{{n / (\ell + 1)}}\right \rfloor\}$.
\end{lemma}

\begin{proof}
For $(G, S_s, S_t, k, \ell)$ a yes-instance of \textsc{VC-R},
there exists a reconfiguration sequence $\alpha$ of length exactly $\ell$ from $S_s$ to $S_t$.
Then by Corollary~\ref{prop-planar-subgraph},
we know that for some fixed $j \in \{0, 1, \ldots, \left \lfloor{{n / (\ell + 1)}}\right \rfloor\}$
we have $V(\alpha) \subseteq V(G^*_j)$ and $V(\alpha) \cap N_G(B^*_j) = \emptyset$,
as otherwise $V(\alpha) \cap B^*_j \neq \emptyset$.
By our construction of $G^*_j$, the maximum capacity constraint is never violated.
Therefore, $\alpha$ is also a reconfiguration sequence from $S^*_{s,j}$ to $S^*_{t,j}$.

For the converse, suppose that $(G^*_j, S^*_{s,j}, S^*_{t,j}, k, \ell)$ is a yes-instance for some fixed
$j \in \{0, 1, \ldots, \left \lfloor{{n / (\ell + 1)}}\right \rfloor\}$ and let $\alpha$
denote the corresponding reconfiguration sequence from $S^*_{s,j}$ to $S^*_{t,j}$.
Since the maximum capacity constraint cannot be violated, we only need to make sure
that (i) no reconfiguration step in $\alpha$ leaves an uncovered edge in $G$ and that (ii) none of
the degree-one gadget vertices are in $V(\alpha)$.
For (ii), it is not hard to see that any such vertex must be touched an even number of
times and we can delete those reconfiguration steps to obtain a shorter reconfiguration sequence.
Moreover, any reconfiguration sequence of length $\ell - x$, where $x$ is even, can be transformed
into a reconfiguration sequence of length $\ell$ by a simple application of the last reconfiguration
step and its reversal $\frac{x}{2}$ times.
For (i), assume that $\alpha$ leaves an uncovered edge in $G$. By our construction of $G^*_j$,
such an edge must have one endpoint in $B^*_j$. But since we added $\ell + 1$ degree-one
neighbors to every vertex in the neighborhood of $B^*_j$, this is not possible.
\qed
\end{proof}

Theorem~\ref{th-vcr-planar} then follows by combining Proposition~\ref{prop-dual},
Lemma~\ref{lem-planar-tw}, Lemma~\ref{lem-vcr-planar}, Theorem~\ref{th-vcr-tw}, and
the fact that $tw(G^*_j) \leq 3\ell - 1$, for all
$j \in \{0, 1, \ldots, \left \lfloor{{n / (\ell + 1)}}\right \rfloor\}$.
\qed
\end{proof}

\subsection{Details omitted from Section~\ref{sec-octr-fvsr-tw}}
\subsubsection{Proof of Theorem~\ref{th-oct-tw}}
\begin{proof}
The proof of correctness proceeds very similarly as for VC-R, we only need to argue that the strengthened last condition for validity (which uses the additional information about bipartitions in an essential way) is now strong enough to carry through the main inductive proof.

\begin{lemma}\label{lem-subsig-valid-oct}
If an OCT-signature $\sig$ over $X$ is obtained from a valid OCT-signature by replacing all vertex steps not in $X$ by $\used$ or $\unused$ steps,
then $\sig$ is valid as well.
\end{lemma}
\begin{proof}
Let $\sig$ be obtained from a valid OCT-signature $\sig'$ over $X'$ by replacing all vertex steps in $X'\setminus X$ by $\used$ or $\unused$ steps.
First note that $\sig(i,S_s \cap X)=\sig'(i,S_s \cap X')\cap X$.
The first three conditions of Definition~\ref{def-valid} follow immediately.
Moreover, if $G[X' \setminus S]$ has a bipartition $L,R$, then $L\cap X,R\cap X$ is a bipartition of $G[X \setminus (S\cap X)]$, hence the fourth condition also follows.
\qed
\end{proof}

\begin{lemma}\label{lem-oct-cover}
Let $G$ be a graph $S,X_1,X_2$ be subsets of $V(G)$ such that every edge of $G[X_1 \cup X_2]$ is contained in $G[X_1]$ or $G[X_2]$.
Let $L,R$ be a partition of $X_1 \cup X_2$.
If $L\cap X_1,R\cap X_1$ is a bipartition of $G[X_1\setminus S]$ and $L\cap X_2,R\cap X_2$ is a bipartition of $G[X_2\setminus S]$, then $L,R$ is a bipartition of $G[(X_1 \cup X_2) \setminus S]$.
\end{lemma}
\begin{proof}
Let $uv$ be an edge of $G[(X_1 \cup X_2)\setminus S]$.
Then it is contained in $G[X_i]$ for some $i\in\{1,2\}$.
It has no endpoint in $S \cap X_i$, hence it is an edge of $G[X_i\setminus S]$.
Thus one endpoint is in $L\cap X_i$ and the other in $R\cap X_i$.
In particular ever edge of $G[(X_1 \cup X_2)\setminus S]$ has one endpoint in $L$ and the other in $R$.
\qed
\end{proof}

\begin{corollary}\label{cor-oct-cover-valid}
Let $\sig,\sig_1,\sig_2$ be a OCT-signatures over $X,X_1,X_2$ respectively, such that $X=X_1\cup X_2$ and every edge of $G[X]$ is contained in $G[X_1]$ or $G[X_2]$.
Assume furthermore that for all $i\leq\ell$:
\begin{itemize}[nosep]
\item[] $\sig[i]=\sig_1[i]$ whenever $\sig[i]\in X_1$ or $\sig_1[i]\in X_1$,
\item[] $\sig[i]=\sig_2[i]$ whenever $\sig[i]\in X_2$ or $\sig_2[i]\in X_2$,
\item[] $\sig[i,v]=\sig_1[i,v]$ whenever $v\in X_1$ and $\sig[i,v]$ is defined and
\item[] $\sig[i,v]=\sig_2[i,v]$ whenever $v\in X_2$ and $\sig[i,v]$ is defined.
\end{itemize}
If $\sig_1$ and $\sig_2$ are valid, then so is $\sig$.
\end{corollary}
\begin{proof}
The assumption implies that $\tau$ and $\tau_1$ agree over all changes within $X_1$, that is, $\tau_1(i,S_s \cap X_1) = \tau(i,S_s \cap X) \cap X_1$ (and similarly for $\tau_2$).
The first three conditions of validity for $\sig$ follow as for VC-R.
For the last condition, it suffices to use Lemma~\ref{lem-oct-cover} for $S=\tau(i,S_s \cap X), L=\{v \mid \sig[i,v]=\ttleft\}, R=\{v \mid \sig[i,v]=\ttright\}$.
\qed
\end{proof}

The following lemma is proved by induction exactly as for VC-R, only with Lemma~\ref{lem-subsig-valid-oct} and Corollary~\ref{cor-oct-cover-valid} used when validity needs to be argued. 

\begin{lemma}\label{lem-dp-induction-oct}
For $i \in V(T)$ and an OCT-signature $\sig$ over $\bag_i$, $\sig\in \table_i$
if and only if
$\sig$ can be extended to an OCT-signature over $V_i$ that is valid.
\end{lemma}

The accepting condition is unchanged and its correctness follows from Lemma~\ref{lem-dp-induction-oct} the same way.
It only remains to consider the running time.
The number of possible OCT-signatures is $(t+3)^\ell 2^{t(\ell-1)}$ (instead of the $(t+3)^\ell$ for VC-R).
In the join operation, we required the new $\sig[i,v]$ entries to be equal and thus the running time is again $2^\ell$ times the number of possible OCT-signatures.
In the forget operation the algorithm only does a polynomial number of calculations for each of the OCT-signatures.
In the introduce operation, for each of the OCT-signatures we consider in the worst case $2^\ell$ possible subsets of $\unused$ steps and $2^\ell$ possible assignments of $\ttleft$ or $\ttright$ to new $\sig[i,v]$ entries.
The total running time is thus $\Ohstar(4^\ell (t+3)^\ell 2^{t\ell})$.

Combining the same complementing technique we used for \textsc{VC-R}
and \textsc{IS-R} with Proposition~\ref{prop-dual}, the result for
\textsc{IBS-R} follows.
\qed
\end{proof}

\end{document}